\documentclass[lettersize,journal]{IEEEtran}
\usepackage{amsmath,amsfonts}
\usepackage{algorithmic}
\usepackage{algorithm}
\usepackage{array}
\usepackage[caption=false,font=normalsize,labelfont=sf,textfont=sf]{subfig}
\usepackage{textcomp}
\usepackage{stfloats}
\usepackage{url}
\usepackage{verbatim}
\usepackage{graphicx}
\usepackage{cite}

\usepackage{comment}
\usepackage{dcolumn}
\newcolumntype{d}[1]{D{.}{.}{#1}}
\usepackage{enumitem}

\usepackage{amssymb}
\usepackage{amsthm}
\usepackage{bm}
\usepackage{float}
\usepackage{enumerate}
\usepackage{mathtools}
\usepackage[colorlinks]{hyperref}       %
\hypersetup{
    colorlinks = false,
    linkcolor = blue,
    anchorcolor = blue,
    citecolor = blue,
    filecolor = blue,
    urlcolor = blue
}
\usepackage{xcolor}         %

\newtheorem{thm}{Theorem}
\newtheorem{lemma}{Lemma}
\newtheorem{cor}{Corollary}
\newtheorem{DF}{Definition}

\newtheorem{rem}{Remark}

\usepackage{pgffor}
\foreach \x in {a,...,z}{%
\expandafter\xdef\csname vec\x \endcsname{\noexpand\ensuremath{\noexpand\bm{\x}}}
}
\foreach \x in {A,...,Z}{%
\expandafter\xdef\csname vec\x \endcsname{\noexpand\ensuremath{\noexpand\bm{\x}}}
}
\foreach \x in {A,...,Z}{%
\expandafter\xdef\csname c\x \endcsname{\noexpand\ensuremath{\noexpand\mathcal{\x}}}
}
\foreach \x in {A,...,Z}{%
\expandafter\xdef\csname bb\x \endcsname{\noexpand\ensuremath{\noexpand\mathbb{\x}}}
}

\newcommand{\defn}{: \, =}

\newcommand{\Lnorm}[2]{\|#1\|_{#2}}

\newcommand{\inprod}[2]{\ensuremath{\langle #1 , \, #2 \rangle}}
\newcommand{\kl}[2]{\ensuremath{D_{\mathsf{KL}}(#1\|#2)}}

\newcommand{\argmin}{\operatornamewithlimits{arg~min}}
\newcommand{\pr}[1]{\mathbb{P}\{ #1 \}}
\newcommand{\var}[1]{\mathsf{Var}( #1 )}
\newcommand{\ind}[1]{\mathbb{I}_{[ #1 ]}}

\newcommand{\Lo}[1]{\|#1\|_{1}}
\newcommand{\Lt}[1]{\|#1\|_{2}}
\newcommand{\Li}[1]{\|#1\|_{\infty}}
\newcommand{\eps}[1]{\epsilon_{1}^{#1}}
\newcommand{\ep}[2]{\epsilon_{#2}^{#1}}
\newcommand{\rn}[1]{r_{1}^{#1}}
\newcommand{\sn}[1]{\mathsf{Sat}(#1)}

\mathtoolsset{showonlyrefs=true}

\begin{document}

\title{Mean Estimation Under Heterogeneous  Privacy Demands}

\author{Syomantak Chaudhuri, Konstantin Miagkov, Thomas A. Courtade \thanks{A preliminary version of this work appeared in IEEE International Symposium of Information Theory 2023 \cite{isit-paper}. \\
S. Chaudhuri and T. A. Courtade are with the Department of Electrical Engineering and Computer Sciences at University of California, Berkeley, CA, USA (email: syomantak,courtade@berkeley.edu). \\
K. Miagkov is with the Department of Mathematics at Stanford University, CA, USA (email: kmiagkov@stanford.edu)

}}

\maketitle

\begin{abstract}
    Differential Privacy (DP) is a well-established  framework to quantify  privacy loss incurred by any algorithm. 
    Traditional formulations impose a uniform privacy requirement for all  users, which is often inconsistent with real-world scenarios in which users dictate their  privacy preferences individually. 
    This work considers the problem of mean estimation, 
     where each user can impose their own distinct privacy level.
    The algorithm we propose is shown to be minimax optimal and has a near-linear run-time.
    Our results elicit an interesting saturation phenomenon that occurs.  
    Namely, the privacy requirements of the most stringent users dictate the overall error rates.
    As a consequence, users with less but differing privacy requirements are all given more privacy than they require, in equal amounts.
    In other words, these privacy-indifferent users are given a nontrivial degree of privacy for free, without any sacrifice in the performance of the estimator.
 
\end{abstract}


\section{Introduction}

Increased computing power and storage technology, combined with fast internet, have made data a precious commodity. Platforms including social media and tech companies have established markets for data that are of great value for targeted advertisements \cite{Data17,Data21}.
In tandem, the ever-growing digital footprint we have has led to a rise in privacy concerns. 
Thus, privacy-preserving techniques in data mining and statistical analysis are important and, at times, mandated by laws such as the GDPR in Europe \cite{GDPR} and the California Consumer Privacy Act (CCPA) \cite{CCPA}.

The study of privacy-preserving techniques for data analysis is an old concern \cite{Hoffman69,Survey00,Rao18}, and simple techniques such as not answering queries specific to a small portion of the dataset do not provide adequate privacy  \cite{Sch75}.
While several notions of privacy have been proposed, such as $K$-Anonymization \cite{Samarati98}, L-diversity \cite{LDiv06}, information-theoretic notions (for example, see \cite{Asoodeh14}), and randomization techniques, the current de-facto standard for privacy -- Differential Privacy (DP) -- was proposed by Dwork et al. \cite{DW06,Dwork06}. 
DP is used in real-world applications by the US Census Bureau \cite{Abowd18}, Google \cite{Erlin14}, and Apple \cite{Tang17}.
One convenient property of DP is that it allows quantification of loss in privacy of an algorithm, as opposed to privacy being a binary property.
Recent extensions of DP include Renyi-DP   \cite{Mironov17}, Concentrated-DP   \cite{Dwork16}, and Zero-Concentrated-DP   \cite{Bun16}.

One of the key tasks in machine learning and statistics is estimating parameters of a distribution given independently drawn samples from it. 
When the observations correspond to sensitive information, such as user data on social media, the need for privacy arises.
Therefore, statistical problems like mean estimation under privacy constraints are important, and we must understand the implicit trade-off between accuracy and privacy. 
The majority of existing  literature on this topic considers a uniform privacy level for all users (see \cite{Wang20} for reference).  However, this does not capture the real-world, where heterogeneous privacy requirements are ubiquitous  (e.g., \cite[Example 2]{Kotsogiannis20}).  Indeed, on virtually all digital platforms, users independently balance their individual privacy options against the utility they desire, typically through a menu of  options provided by the platform.

\subsection{Our Contribution}
We consider  mean estimation  under the Central-DP (\textit{CDP}) model, with heterogeneous privacy demands. 
In the CDP model, also known as the Trusted-Curator model, users send their true data to a central server which is expected to respect the privacy constraint set forth by the users \cite{DW06,Dwork06}.  We assume each user has a datapoint sampled from an unknown distribution, and users are allowed to select their own privacy level. 
The class of distributions is assumed to be univariate and bounded in a known range, and we consider the minimax expected squared error as our metric.  We propose a certain affine estimator with judiciously chosen weights, and prove it to be minimax optimal.
Our algorithm for computing the weights is efficient and has a near-linear (in the number of users) run-time and linear space complexity.

As is  the case in homogeneous DP, keeping the privacy requirement of some users fixed, one might expect to get better accuracy in mean estimation as the privacy of the other users are relaxed.
However, we show that after a certain critical value, decreasing the privacy provides no further improvement in the accuracy of our estimator.
By matching upper and lower bounds, we show that this phenomenon is fundamental to the  problem itself and not an artifact of our algorithm.
As a corollary of this saturation phenomenon, having an additional public dataset may have no extra benefit for mean estimation than a private dataset.
Thus, the central-server can advertise and offer some extra privacy up to the critical value to the privacy-indifferent users while not sacrificing the estimation performance. 
Experiments confirm the superior performance of our proposed algorithm over other baseline methods.
In addition, the approach for showing that the upper and the lower bounds are within constant factor of each other may be of independent interest to the readers.

\subsubsection*{Organization}
In Section~\ref{sec:PD}, we define the problem setting.
The proposed algorithm along with the upper bounds are presented in Section~\ref{sec:UB}, and the lower bound is presented in Section~\ref{sec:LB}.
The fact that the lower and the upper bounds are within a constant factor of each other is proved in Section~\ref{sec:OPT}.
Experiments and other baseline methods are presented in Section~\ref{sec:Exp} to support the theoretical claims made in this work. 
Conclusions and possible future directions are outlined in Section~\ref{sec:Con}.

\subsection{Related Work}
Estimation error in the homogeneous DP case has been studied in great detail in the recent years (see \cite{Kamath19,Hopkins22,Kamath20P,Cai19}) under both the CDP model and the Local-DP (\textit{LDP}) model. 
In the LDP model, users do not trust the central server and send their data through a noisy channel to the server to preserve privacy \cite{Kasiviswanathan11,Duchi16}. 

Tasks like query release, estimation, learning, and optimization have been considered in the setting of a private dataset assisted by some public data \cite{Bassily20,Bassily20-learn,Liu21,Alon19,Nandi20,Kairouz21,Amid22,Wang19}. 
An importance sampling scheme to release statistical queries on a private dataset using a public dataset based on logistic regression is proposed in \cite{Ji13}. 
This method is unsuitable for our setting, where all users have the same underlying distribution.
Bie et al. \cite{Kamath22} consider using a few public samples to estimate Gaussian distributions with unknown mean and covariance matrix. 
The public samples eliminate the need for prior knowledge on the range of mean, but the effect on accuracy with more public samples is not studied.
One-Sided DP to combat `exclusion attacks' to protect a private dataset in the presence of a public dataset is considered in \cite{Kotsogiannis20}.
Another form of heterogeneity is a hybrid model where some users are satisfied with the CDP model while other users prefer the LDP model \cite{Avent17,Beimel19}.
Heterogeneous DP \textit{(HDP)} for federated learning is considered in \cite{Liu21Pj,Aldaghri21}. 
Alaggan et al. \cite{Alaggan17} give a general recipe for dealing with HDP but their idea of scaling the data using a shrinkage matrix induces a bias in the estimator. Further, their approach can not deal with public datasets.

Personalized Differential Privacy (\textit{PDP}) is another term for HDP in literature. Li et al. \cite{Li17Par} studied PDP and proposed a computationally expensive way to partition users into groups with similar privacy levels. 
Jorgensen et al. \cite{Jorg15} propose a mechanism that samples users with high privacy requirements with less probability and then the sub-sampled dataset can then be analyzed with the conventional homogeneous DP algorithms under a suitable privacy level.
While this is a general approach for dealing with heterogeneity, it is not optimal for mean estimation. 
Indeed, following the recommendation of the author for setting the parameters of the algorithm, in the presence of some public data in the dataset, the mechanism would ignore all private data.
HDP mean estimation under the assumption that the variance of the unknown distribution is known is considered by Ferrando et al. \cite{Ferrando21}.
However, as they mention, they add more noise than necessary for privacy. 
The reason why they add more noise than required is that they are essentially performing LDP instead of the more powerful CDP technique. 
As a result, no saturation phenomenon can be deduced in their method.
Some works \cite{Jorg15,Niu20} also consider the PDP setting for finite sets and give algorithms inspired by the Exponential Mechanism \cite{McSherry07}. 
Heterogeneous privacy problems for recommendation systems is also considered in \cite{Li17CF,Zhang19}.
 PDP in the LDP setting has been studied by \cite{Chen16} for learning the locations of users from a finite set of possible locations.
More general notions of DP which encompasses HDP have also been considered in literature \cite{Torkamani22,Chatzikokolakis13}. 
Another line of work in literature considers DP under heterogeneity in the data \cite{Noble22,Gupta22}.
Mean estimation, in the Bayesian setting, from user data in a private manner is considered in \cite{Cum22} where the heterogeneity is in the number of samples and distribution  of each user's data, and not in the desired privacy levels.

Most closely related to the present work is that of Fallah et al. \cite{Asu22}, which considers the general HDP setting for mean estimation in the context of efficient auction mechanism design from a Bayesian perspective. 
While they encounter a saturation-type phenomenon in their algorithm, it cannot tightly characterize the saturation condition (see Section~\ref{sec:Exp}).
In addition, for the homogeneous case, their algorithm does not agree with the minimax optimal estimator when they privacy demand is low ($\epsilon > 1/\sqrt{n}$ where $n$ is the total number of users and $\epsilon$ is the DP privacy-level defined in the next Section). 
They also assume that all the privacy levels are less than $1$. 
This assumption is central to their upper and lower bounds, and therefore one cannot draw conclusions for the case when there is a public dataset.  
Section~\ref{sec:Exp} contains more comparisons of our proposed method with that of \cite{Asu22}.

Concurrent with our work \cite{isit-paper}, 
Cummings et al.~\cite[Section 6]{Cummings23} consider the same optimization problem that we solve to obtain an affine estimator of the mean.
They also make observations related to  the solution structure which are similar to ours (Lemma~\ref{lem:domain}), including the saturation phenomenon. 
In their solution, to find the optimal weights and the optimal noise level, one needs to know the index of saturation in their Claim 7.
The method they prescribe is via enumerating the $n+1$ indices, where $n$ is the total number of users. 
While they point out that this requires $O(n)$ time, this does not provide a full accounting of the complexity for two reasons. 
First, sorting the privacy levels of the agents requires $n\log n$ time.
Second, their method makes $O(n)$ calls to a function that requires $O(n)$ compute time, rendering an overall complexity of $O(n^2)$.
ADPM, our proposed solution, has a $O(n \log n)$ runtime ($O(n)$ if we don't consider sorting complexity).
We are able to do this since  ADPM finds the index of saturation naturally without a brute-force search.
Further, the approach we use is different from theirs and provides some more insight on how the weights depend on the privacy parameters. 
Finally, we also prove that our method is within a constant factor of the lower bound, which establishes optimality; lower bounds are not considered in Cummings et al.~\cite{Cummings23}.

\section{Problem Definition and Main Result} \label{sec:PD}

\subsection{Problem Definition}
Heterogeneous Differential Privacy (HDP)  allows different users to have different possible privacy requirements. 
We begin by defining a natural notion of heterogeneous Differential Privacy. 
Similar definitions were also considered in \cite{Asu22,Alaggan17}. 
We denote positive real numbers by $\bbR_{> 0}$. 
As we consider univariate data-points in our datasets, we use boldface, such as $\vecx$ to denote a dataset, or equivalently, a vector. 
Capital boldface, such as $\vecX$, are used to denote a random dataset, i.e., a random vector. 
Vectors with subscript $i$, e.g. $\vecx_i$, refer to the $i$-th entry of the vector, while we use the notion $\vecx'^i$ for a vector differing from $\vecx$ at the $i$-th position.
The notation $[n]$ refers to the set $\{1,2,\ldots,n\}$.
The probability simplex in $n$-dimensions is represented by $\Delta_n$.
Throughout this work the $\ell^1$ norm, $\Lo{\cdot}$, is interchangeably used for sum of elements of vectors with positive components.
The notation $a \wedge b$ is used to denote $\min\{a,b\}$.
In this work, with some abuse of notation, we shall represent a \textit{sample} from a randomized algorithm $M$ mapping $\cX^n$ to a probability distribution on ${\cY}$ as $M(\vecx)$ where $\vecx \in \cX^n$.

\begin{DF}[Heterogeneous Differential Privacy] \label{def:epsDP}
A randomized algorithm $M: \cX^n \to \cY$ is said to be $\bm\epsilon$-DP for $\bm\epsilon \in \bbR_{> 0}^n$ if 
\begin{equation} \label{eq:DP-def}
    \pr{M(\vecx) \in S} \leq e^{\epsilon_i} \pr{M(\vecx'^i) \in S} \ \ \  \forall i \in [n],
\end{equation}
for all measurable sets $S \subseteq \cY$, where $\vecx,\vecx'^i \in \cX^n$ are any two `neighboring' datasets that differ arbitrarily in only the $i$-th component. 
\end{DF}

\begin{rem}
Note that the probability in the above definition is taken over the randomized algorithm conditioned on the given datasets $\vecx,\vecx'^i$, i.e., it is a conditional probability.
\end{rem}

Without loss of generality, we consider the case $\cX = [-0.5,0.5]$ and let $\cP$ denote the set of all distributions with support on $\cX$. 
Our results can be directly extended to distributions on any known finite length domain $\cX$.
Under this privacy setting, we investigate the problem of estimating the sample mean from the users' data.
There are $n$ users and each user's data point is sampled i.i.d.~from a distribution $P \in \cP$ over $\cX$ with mean denoted by $\mu_P \in [-0.5,0.5]$ henceforth. 
Each `data point' corresponds to a user's data in $\cX$, i.e., user $i$ has a datapoint (sample) $\vecx_i$ and the user has a privacy requirement  of $\epsilon_i$ (in the sense of Definition \ref{def:epsDP}).
We assume that the privacy constraint $\epsilon_i$ does not depend on the realization $\vecx_i$, and is itself not private.  %
Without loss of generality, we assume that the vector of privacy levels $\bm\epsilon$ is arranged in a non-decreasing order.

Let the set of all $\bm\epsilon$-DP algorithms from $\cX^n$ to $\cY = [-0.5,0.5]$ be denoted by $\cM_{\bm\epsilon}$. 
We consider the error metric as Mean-Squared Error (MSE) and are interested in characterizing the minimax estimation error, over all $\bm\epsilon$-DP algorithms. 
For an algorithm $M(\cdot) \in \cM_{\bm\epsilon}$, let $E(M)$ denote the worst-case error attained by it,
\begin{equation}
    E(M) := \max_{P \in \cP}\ \bbE_{\vecX\sim P^n,M(\cdot)}[(M(\vecX)-\mu_P)^2]. \label{eq:EM}
\end{equation}
In \eqref{eq:EM}, the expectation is taken over the randomness in the dataset $\vecX$ and the algorithm $M(\cdot)$.
Let $L(\bm\epsilon)$ denote the minimax estimation error given by
\begin{equation} \label{eq:minimax-def}
 L(\bm\epsilon ) \defn \min_{M \in \cM_{\bm\epsilon}}  E(M).
\end{equation}

Our goal is to characterize $L(\bm\epsilon)$ and provide an algorithm that achieves a MSE of the same order.

\subsection{Main Result}

\begin{algorithm}[!t]
\caption{Affine Differentially Private Mean (ADPM)}\
\label{alg:ADPM}
\begin{algorithmic}
\STATE \hspace{-0.2cm}{\textsc{ADPM}}$(\bm\epsilon,\vecx)$
\STATE $n \gets \textsc{length}(\bm\epsilon)$
\STATE $\bm\epsilon \gets $\textsc{sort}($\bm\epsilon$) \hspace{0.6cm}(ascending order)
\STATE $r_1 \gets \epsilon_1$
\STATE $L_1 \gets \epsilon_1$
\STATE $L_2 \gets \epsilon_1^2$
\STATE $k \gets 1$
 \WHILE{$k < n$}
     \STATE $r_{k+1} \gets \min\{\epsilon_{k+1} , \frac{L_2 + 8}{L_1}\}$ 
     \STATE $L_1 \gets L_1 + r_{k+1}$
     \STATE $L_2 \gets L_2 + r_{k+1}^2$
     \STATE $k \gets k + 1$
\ENDWHILE
\IF{$\frac{L_2 +8}{4L_1^2} > \frac{1}{4}$}
\STATE \textbf{return} $0$
\ELSE
\STATE $\vecw \gets  \vecr/L_1$
\STATE sample $N \sim $ Laplace($1/L_1$)
\STATE \textbf{return} $\inprod{\vecw}{\vecx} + N$
\ENDIF
\end{algorithmic}
\end{algorithm}

Our main result is that our proposed algorithm ADPM (Algorithm~\ref{alg:ADPM}) is minimax optimal, as stated  in Theorem~\ref{thm:main} below.
Note that Theorem~\ref{thm:main} is instance optimal, in the sense that it establishes  minimax optimality for each $n$ and $(\epsilon_1,\ldots,\epsilon_n)$.

ADPM has a near-linear time and linear space complexity.  More precisely, the initial sorting of $\bm\epsilon$ requires $O(n \log n)$ time and $O(n)$ space, 
computing the weights and the inner product $\vecw$ requires $O(n)$ time and space each.
Note that for non-private mean estimation, we still require $O(n)$ time to compute the mean.
For a special setting, if we know that there are $n$ users with any user's privacy parameter taking values in a known discrete set of values $\cE$ of size $|\cE| = k$, ($k < n$) then computing the weights can be done in $O(k \log k )$ time instead.

\begin{thm} \label{thm:main}
    ADPM described in Algorithm~\ref{alg:ADPM} achieves worst-case error within a (universal) constant factor of $L(\bm\epsilon)$, for all $n$ and privacy constraints $\bm\epsilon = (\epsilon_1,\ldots,\epsilon_n)$.
\end{thm}

In Theorem~\ref{thm:UB} (Section~\ref{sec:UB}), we prove an upper bound on the MSE for ADPM, which is an upper bound on $L(\bm\epsilon)$. 
A lower bound on $L(\bm\epsilon)$ is shown in Theorem~\ref{thm:lb} (Section~\ref{sec:LB}).
Due to the  different form of the lower and upper bounds, it is non-trivial to compare them. 
In Theorem~\ref{thm:opt} (Section~\ref{sec:OPT}), it is shown that the lower and upper bound are within a constant factor of each other, proving the minimax optimality of ADPM and thus, proving Theorem~\ref{thm:main}.

We shall switch to working with variable length vectors at times so we define some new notation.
Consider any sequence of non-decreasing privacy values $\{\epsilon_i\}_{i=1}^{n}$.
The notation $\ep{j}{i}$ refers to the vector $(\epsilon_i,\ldots,\epsilon_j)$ for $j \geq i$.

We now describe an important phenomenon in ADPM, that, by minimax optimality of ADPM implied by Theorem~\ref{thm:main}, is fundamental to the problem.
Assume $\bm\epsilon$ to be in non-decreasing order, i.e., the users with stricter privacy requirements are arranged earlier in the vector $\bm\epsilon$.
Let $k$ be the minimum index at which $\epsilon_{k+1} \geq \frac{\Lt{\eps{k}}^2+8}{\Lo{\eps{k}}}$.
Note that such a $k$ need not exist if all the privacy levels are sufficiently close to each other.
When such a $k$ exists, ADPM algorithm provides a privacy level exactly  $\epsilon_i$ to user $i$ for $i \in [k]$ and for the rest of the users,
it provides a common privacy level of $\frac{\Lt{\eps{k}}^2+8}{\Lo{\eps{k}}}$.
Thus, for these latter users with lower privacy requirements, ADPM offers extra privacy without sacrificing on  MSE.
It is interesting to note that among the latter users, regardless of their relative privacy requirements, they all receive the same level of privacy.
Thus, even if there are some users who are do not care about privacy ($\epsilon_i \to \infty$), it is still optimal to give them some privacy.

As an example, if out of the $n \ (> 10^3)$ users, there are $10^3$ users wanting a privacy requirement of $\epsilon = 0.1$ and the rest of the users have privacy requirements ranging from $\epsilon = 0.5$ to $\epsilon \to \infty$, then all the users in the latter category receive a privacy guarantee of $\epsilon = 0.18$ by ADPM.
\section{Upper Bound} \label{sec:UB}

The main result of this section, an upper bound on the performance of ADPM (and therefore $L(\bm\epsilon)$) is stated in  Theorem~\ref{thm:UB}. 
Subsequently, we motivate ADPM, observe some properties and then prove the theorem at the end of this section.

\begin{thm}[Upper Bound] \label{thm:UB}
Let $r_1 = \epsilon_1$ and 
\begin{equation}
    r_{k+1} = \min\left\{\epsilon_{k+1}, \frac{\Lt{\rn{k}}^2 + 8}{\Lo{\rn{k}}} \right\}\ \forall \ k \in [n-1]. \label{eq:thm2-r-def}
\end{equation}  
Then, ADPM defined in Algorithm~\ref{alg:ADPM}, has a worst case MSE of $\frac{\Lt{\rn{n}}^2 + 8}{4\Lo{\rn{n}}^2} \wedge \frac{1}{4}$, and thus,
\begin{equation}
   L(\bm\epsilon) \leq \frac{\Lt{\rn{n}}^2 + 8}{4\Lo{\rn{n}}^2} \wedge \frac{1}{4}.
\end{equation}  
\end{thm}

The upper-bound, given by the MSE of ADPM, is centered around finding the optimal affine estimator $\inprod{\vecw}{\vecx} + N$ of the mean, where $\vecw \in \Delta_n$ and $N$ is some suitable noise to satisfy the $\bm\epsilon$-DP constraint.
In particular, we prove that $\vecw = \vecr/\Lo{\vecr}$, for the vector $\vecr$ recursively defined in \eqref{eq:thm2-r-def}, is a global-minimizer of the MSE.

\subsection{ADPM Motivation}

Let us constrain ourselves to affine estimators of form $M(\vecx) = \inprod{\vecw}{\vecx} + N$, where a zero-mean random noise $N$ is chosen appropriately to satisfy the privacy constraint and  $\vecw \in \Delta_n$ to make the estimator unbiased.
Recall that Laplace$(\eta)$ distribution has pdf   $\frac{1}{2\eta}e^{-|\cdot|/\eta}$.
If the noise is distributed as Laplace$(\eta)$, then it can be shown that the estimator $\inprod{\vecw}{\vecx} + N$, is $(\vecw/\eta)$-DP (see Lemma~\ref{lem:eps-af} in Appendix~\ref{ap:A}).
Thus, from the privacy constraint, we impose the condition
\begin{equation} \label{eq:w-con}
    w_i \leq \eta \epsilon_i \ \forall i.
\end{equation}
The variance (or MSE) of the estimator under distribution $P$ is given by $\var{P}\Lnorm{w}{2}^2 + 2\eta^2 \leq \Lnorm{w}{2}^2/4 + 2\eta^2$ ($1/4$ is the worst case variance for distributions on $\cP$).
To minimize this, subject to \eqref{eq:w-con}, we set $\eta = \max_i w_i/\epsilon_i = \Li{\frac{\vecw}{\bm\epsilon}}$, where the latter notation is element-wise division.
Therefore, we have
\begin{equation}
    \text{MSE} \leq \frac{\Lt{\vecw}^2}{4} + 2\left\|\frac{\vecw}{\bm\epsilon} \right\|^2_{\infty}. \label{eq:pd-val}
\end{equation}
Thus, finding the optimal affine estimator requires us to solve the minimization problem
\begin{equation} \label{eq:pd-opt}
    \vecw^* = \argmin_{\vecw \in \Delta_n} \frac{\Lt{\vecw}^2}{4} + 2\left\|\frac{\vecw}{\bm\epsilon} \right\|^2_{\infty}.
\end{equation}

Although \eqref{eq:pd-opt} can be solved by any modern convex optimization solver, we need to solve it analytically to get an upper bound on the MSE for showing minimax optimality.
In the next subsection, we show how to solve this optimization problem in a recursive manner and ADPM uses this solution.

\begin{rem}[Sub-Optimality of Proportional Weighing] \label{rem:prop}
Note that \eqref{eq:w-con} can be satisfied by taking $\vecw \propto \bm\epsilon$, i.e., $\vecw = \bm\epsilon/\Lo{\bm\epsilon}$ and $\eta = 1/\Lo{\bm\epsilon}$.
This corresponds to providing users with higher privacy a lower weightage in the estimator. 
While intuitive, this is not optimal.
Consider the case where there are a total of $1000$ users out of which $999$ demand a privacy requirement $\epsilon = 0.1$ and one user has no privacy requirement ($\epsilon \to \infty$).
In this case, the above estimator just considers a single data point and has worst case MSE of $\frac{1}{4}$.
One can do better by using the weights $\vecw^*$ in \eqref{eq:pd-opt}, which would give a worst case MSE of the order $10^{-4}$.
\end{rem}

\subsection{Solving the Minimization Problem}

We use the function $f(\vecx,\bm\epsilon)$\footnote{$\vecx$ from here on doesn't refer to the dataset.} to denote the upper bound on the MSE when using the weights $\vecx \in \mathbb{R}_{\geq 0}^n$ for privacy requirement $\bm\epsilon$ (see \eqref{eq:pd-val}).
For convenience, we do not restrict the weights $\vecx$ to be on the simplex and instead, we scale it to the simplex by considering $\vecw = \frac{\vecx}{\Lo{\vecx}}$ to be the actual weights, i.e.,
\begin{equation}
    f(\vecx,\bm\epsilon) = \frac{\Lt{\vecx}^2}{4\Lo{\vecx}^2} + 2\frac{\left\|\vecx/\bm\epsilon\right\|^2_{\infty}}{\Lo{\vecx}^2}.
\end{equation}

The reason we do not restrict $\vecx$ to the simplex is that it is easier to work with linearly-scaled domain instead of the simplex due to a pattern that emerges in the solution that we describe in this subsection.

We remind the readers that $\bm\epsilon$ is arranged in a non-decreasing order, i.e., $\epsilon_i \leq \epsilon_j,\ \forall i < j$.
With this in mind, Lemma~\ref{lem:domain} shows an important property of the optimal weights $\vecw^*$ (in the simplex) for the minimization problem under consideration.

\begin{lemma} \label{lem:domain}
    Consider a fixed length non-decreasing privacy constraint vector $\bm\epsilon$.
    Let $\vecw^* = \argmin_{\vecw \in \Delta_n} f(\vecw,\bm\epsilon)$, then 
    \begin{equation}
        \frac{w^*_i}{\epsilon_i} \geq \frac{w^*_j}{\epsilon_j} \quad \forall i < j.
    \end{equation}
\end{lemma}
\begin{proof}
We show that for any $\vecw \in \Delta_n$ such that 
$\exists i < j$, $\frac{w_i}{\epsilon_i} < \frac{w_j}{\epsilon_j}$,
     it is possible to find a $\tilde\vecw \in \Delta_n$ such that $f(\tilde\vecw,\bm\epsilon) < f(\vecw,\bm\epsilon)$.
    Let $\frac{w_j}{w_i} = K$, and $\lambda = \frac{\epsilon_j}{\epsilon_i}$.
    Hence, $K > \lambda \geq 1$.
    Consider $\tilde \vecw$ that is equal to $\vecw$ except $\tilde w_j = w_i(K - \delta)$ and  $\tilde w_i = w_i(1 + \delta)$.
    Thus, $\Lo{\tilde \vecw} = \Lo{\vecw}$
    Choosing any $0 < \delta \leq \frac{K - \lambda}{1+ \lambda}$ implies $\frac{w_j^*}{\epsilon_j} > \frac{\tilde w_j}{\epsilon_j} \geq \frac{\tilde w_i}{\epsilon_i} > \frac{w_i}{\epsilon_i}$.
    Thus, $\Li{\tilde\vecw/\bm\epsilon} \leq \Li{\vecw/\bm\epsilon}$.

    Next, $(w_j^2 + w_i^2) - (\tilde w_j^2 + \tilde w_i^2) = 2\delta w_i^2 (K-1-\delta)$. Thus, choosing $0 < \delta < K-1$, we have $\Lt{\vecw}^2 > \Lt{\tilde \vecw}^2$.

    Thus, overall, any $0 < \delta < \frac{K - \lambda}{1+ \lambda}$ would result in $f(\tilde\vecw,\bm\epsilon) < f(\vecw,\bm\epsilon)$.
\end{proof}

We shall instead minimize (globally) $f(\vecx,\bm\epsilon)$ over $\vecx$ with the constraint $x_1 = \epsilon_1$.  
Since the optimal weights 
$$\vecr = \argmin_{\vecx \in \mathbb{R}_{\geq 0}^n: x_1 = \epsilon_1} f(\vecx,\bm\epsilon)$$ 
is given by the linear-scaling\footnote{Lemma~\ref{lem:domain} also implies that $w_1^*$ is strictly greater than $0$} $r_i = \left( \frac{\epsilon_1}{w_1^*}\right) w_i^*$, it satisfies Lemma~\ref{lem:domain} as well, i.e.,
$$ \frac{r_i}{\epsilon_i} \geq \frac{r_j}{\epsilon_j} \ \forall i \leq j.$$
This allows us to constrain our search to the global optimizer to a smaller region. 
Thus, it is sufficient to perform a constrained minimization over the domain 
\begin{equation}
    \cD(\bm\epsilon) = \left\{\vecx \in \bbR^n_{\geq 0}: x_1 = \epsilon_1,\  \frac{x_i}{\epsilon_i} \geq \frac{x_j}{\epsilon_j} \ \forall j>i \right\}. 
\end{equation}
Note that this is a closed convex domain. 
Further, constraining to this domain allows us to write $\Li{\frac{\vecx}{\bm\epsilon}} = 1 \ \forall \ \vecx \in \cD(\bm\epsilon)$.
Thus, this observation and Lemma~\ref{lem:domain} allows us to get Corollary~\ref{cor:opt}.

\begin{cor} \label{cor:opt}
We note that
\begin{equation}
\vecr = \argmin_{\vecx \in \cD(\bm\epsilon)} f(\vecx,\bm\epsilon) =  \argmin_{\vecx \in \cD(\bm\epsilon)} \frac{\Lt{\vecx}^2 + 8}{4\Lo{\vecx}^2}, \label{eq:rn-defn}    
\end{equation}
and
\begin{equation}
     \argmin_{\vecw \in \Delta_n} f(\vecw,\bm\epsilon) = \frac{\vecr}{\Lo{\vecr}}.
\end{equation}
\end{cor}

In Lemma~\ref{lem:quasi}, we show that the function $\frac{\Lt{\vecx}^2 + 8}{4\Lo{\vecx}^2}$ is strictly quasi-convex, which means that a local minimizer in \eqref{eq:rn-defn} is  the unique global minimizer.

\begin{lemma} \label{lem:quasi}
    The function 
    \begin{equation}
    K(\vecx) = \frac{\Lt{\vecx}^2+8}{\Lo{\vecx}^2}    
    \end{equation}
    is strictly quasi-convex (domain restricted to the non-negative reals and  $\Lo{\vecx} > 0$).
\end{lemma}
\begin{proof}
    Since the domain is non-negative reals, we use $\Lo{\cdot}$ to represent sums.
    Consider $\vecx,\vecy$ in domain and $\vecx \neq \vecy$. Let $\lambda \in (0,1)$ then,
    \begin{align}
        K(\lambda &\vecx + (1-\lambda) \vecy) =\\
         &\frac{\lambda^2\Lt{\vecx}^2+(1-\lambda)^2\Lt{\vecy}^2 + 2\lambda(1-\lambda)\inprod{\vecx}{\vecy}+8}{\lambda^2\Lo{\vecx}^2 + (1-\lambda)^2\Lo{\vecy}^2 + 2\lambda(1-\lambda)\Lo{\vecx}\Lo{\vecy}} \\
         &\leq \max\left\{K(\vecx),K(\vecy), \frac{\inprod{\vecx}{\vecy}+8}{\Lo{\vecx}\Lo{\vecy}}\right\}
    \end{align}
    In the above, for equality to hold, it is necessary that $K(\vecx) = K(\vecy) = \frac{\inprod{\vecx}{\vecy}+8}{\Lo{\vecx}\Lo{\vecy}}$.
    We show that this is not possible, implying strict quasi-convexity.
    
    To show $\frac{\inprod{\vecx}{\vecy}+8}{\Lo{\vecx}\Lo{\vecy}} < \max\{K(\vecx),K(\vecy)\}$,
    let $\vecx,\vecy$ be $n$-dimensional  vectors. 
    Then, in the following, consider $\vecx',\vecy'$ to be the $(n+1)$-dimensional vectors obtained by appending a $\sqrt{8}$ at the end to $\vecx,\vecy$ respectively.
    We have,
    \begin{align}
    \frac{\inprod{\vecx}{\vecy}+8}{\Lo{\vecx}\Lo{\vecy}} &= \frac{\inprod{\vecx'}{\vecy'}}{\Lo{\vecx}\Lo{\vecy}} \\
    &< \frac{\Lt{\vecx'}\Lt{\vecy'}}{\Lo{\vecx}\Lo{\vecy}} \quad\quad  (\vecx' \neq c\vecy') \label{eq:CS-ineq}\\
    &\leq \max\left\{ \frac{\Lt{\vecx'}^2}{\Lo{\vecx}^2}, \frac{\Lt{\vecy'}^2}{\Lo{\vecy}^2} \right\}\\
    &= \max\left\{ K(\vecx),K(\vecy) \right\}.
    \end{align}
    where \eqref{eq:CS-ineq} follows from Cauchy-Schwarz inequality and $\vecx'$ can not be proportional to $\vecy'$ since they have equal values in the $(n+1)$-th coordinate and $\vecx \neq \vecy$.
\end{proof}

We shall now overload the the function $f(\cdot,\cdot)$ to map variable length vectors to $\bbR$.
Corresponding to the $\epsilon$ series, we construct a series of weights denoted by $\{r_i\}_{i=1}^{n}$ which is the solution to \eqref{eq:rn-defn}.
Lemma~\ref{lem:quasi} implies that it is sufficient to search for a local minimizer in \eqref{eq:rn-defn}. 
We can find the local minimizer recursively.
We drop the privacy argument in $f(\cdot,\cdot)$ for convenience in the following writing,
i.e., $f(x_1^k)$ refers to $f(x_1^k,\eps{k}) =  \frac{\Lt{x_1^k}^2 + 8}{4\Lo{x_1^k}^2}$ (when $x_1^k \in \cD(\eps{k})$).

Next, we recursively build a sequence of scaled weights that gives the solution to \eqref{eq:rn-defn}.
We remind the readers that the $\{\epsilon_i\}_1^n$ sequence is arranged in a non-decreasing order.
Let $r_1 = \epsilon_1$ and define
    \begin{equation} \label{eq:r-def}
    r_{k+1} = \min \left\{ \epsilon_{k+1} ,  4f(\rn{k})\Lo{\rn{k}} \right\}.
\end{equation}

For clarity, $4f(\rn{k})\Lo{\rn{k}} = \frac{\Lt{\rn{k}}^2 + 8}{\Lo{\rn{k}}}$, i.e., we have $\Li{\frac{\rn{k}}{\eps{k}}} =1 $; this will be clear from Lemma~\ref{lem:sat}.
Algorithm~\ref{alg:ADPM} is an efficient implementation of the recursion described above in \eqref{eq:r-def} that runs in near-linear time.
We shall prove that this sequence generates the optimal weights $\rn{j}$ for privacy constraint $\eps{j}$ for any $j>0$.
An intuition behind the particular recursive definition is presented in Section~\ref{sec:ADPM-inter}.
A notion of saturation is defined in Definition~\ref{def:sat} which is useful in tracking some properties of the $\{r_i\}_1^n$ sequence presented in Lemma~\ref{lem:sat}.  Below, $\mathbb{I}$  denotes the indicator function.

\begin{DF}[Saturation] \label{def:sat}
For a fixed $\eps{n}$, construct $\rn{n}$ as described in \eqref{eq:r-def}.
Consider the event $\sn{k+1} = \ind{r_{k+1} = 4f(\rn{k})\Lo{\rn{k}}}$, i.e.,
the event $\sn{k+1}$ occurs when the minimum term in \eqref{eq:r-def} is
$4f(\rn{k})\Lo{\rn{k}}$. 
\end{DF}

We call this saturation since the weight assigned to the new datapoint is
independent of $\epsilon_{k+1}$; this occurs when $\epsilon_{k+1}$ is too large and assigning proportional weight is sub-optimal as suggested in Remark~\ref{rem:prop}.
Lemma~\ref{lem:sat}(A) shows that once there is saturation at some index, all the following indices are saturated as well. 
Lemma~\ref{lem:sat}(B)  also allows us to understand how fast the MSE decreases after saturation occurs, which will be used in Theorem~\ref{thm:opt} later.
Note that Lemma~\ref{lem:sat}(C) can be used to verify that the $\{r_i\}$ sequence  satisfies $\rn{k} \in \cD(\eps{k}) \ \forall k$  \footnote{Thus, $\Li{\frac{\rn{k}}{\eps{k}}} =1 $.} and hence, the sequence is in the domain of the optimization in \eqref{eq:rn-defn}. 
\begin{lemma} \label{lem:sat}
    If $\sn{j+1}$ occurs, then we have the following \\
    A) $\bigwedge_{i=j+1}^{n} \sn{i}$ occurs, \\
    B) $f(\rn{j+1})=f(\rn{j})/(1+4f(\rn{j}))$, \\
    C) $r_{j+1} = r_{j+2} =  \ldots = r_n$ . \\
    Further, we have 
    \begin{equation}
        r_i \leq r_k, \text{ for any } i < k. \label{eq:r-inc}
    \end{equation}
\end{lemma}
\begin{proof}
If $\sn{j+1}$, we have
    \begin{align}
        4f(\rn{j})\Lo{\rn{j}} &\leq \epsilon_{j+1}, \\
        r_{j+1} &= 4f(\rn{j})\Lo{\rn{j}}.
    \end{align}
    Evaluating $f(\rn{j+1})$ using this, 
    \begin{align}
        f(\rn{j+1}) &= \frac{\Lt{\rn{j}}^2 + r_{j+1}^2 + 8}
        {4\Lo{r_1^j}^2(1+4f(\rn{j}))^2} \\
        &= \frac{4\Lo{r_1^j}^2f(\rn{j}) + r_{j+1}^2}
        {4\Lo{r_1^j}^2(1+4f(\rn{j}))^2} \\
        &= \frac{4\Lo{r_1^j}^2f(\rn{j}) + 4^2f(\rn{j})^2\Lo{\rn{j}}^2}
        {4\Lo{r_1^j}^2(1+4f(\rn{j}))^2} \\
        &= f(\rn{j})/(1+4f(\rn{j})).
    \end{align}
    This proves (B). 
    
    To check for $\sn{j+2}$, compare 
    $\epsilon_{j+2}$ and $4f(\rn{j+1})\Lo{\rn{j+1}}$.
    Note that $\epsilon_{j+2} \geq \epsilon_{j+1}$ since $\epsilon$-sequence is non-decreasing.
    Further, $\Lo{\rn{j+1}} = \Lo{\rn{j}}(1+4f(\rn{j}))$ which implies $4f(\rn{j+1})\Lo{\rn{j+1}} = 4f(\rn{j})\Lo{\rn{j}}$ by (B).
    Therefore, $\sn{j+2}$ occurs as well and (A) is proved by repeating the same argument for $\sn{j+3},\ldots,\sn{n}$.
    
    Further, as noted, $4f(\rn{j+1})\Lo{\rn{j+1}} = 4f(\rn{j})\Lo{\rn{j}}$ when $\sn{j+1}$. 
    Along with (A) and the fact that $\epsilon_k$ are increasing implies $r_{j+1} = r_{j+2}$. 
    Repeating the same argument proves (C).
    
    $r_i \leq r_k$ for any $i < k$ follows immediately: if saturation does not occur until index $j$, then $r_1^{j-1} = \eps{j-1}$, which is non-decreasing. 
    Once saturation happens at any index (not necessary to occur), then the $r$ values stay constant by part (C).
\end{proof}

\begin{figure}
    \centering
    \includegraphics[width=0.5\textwidth]{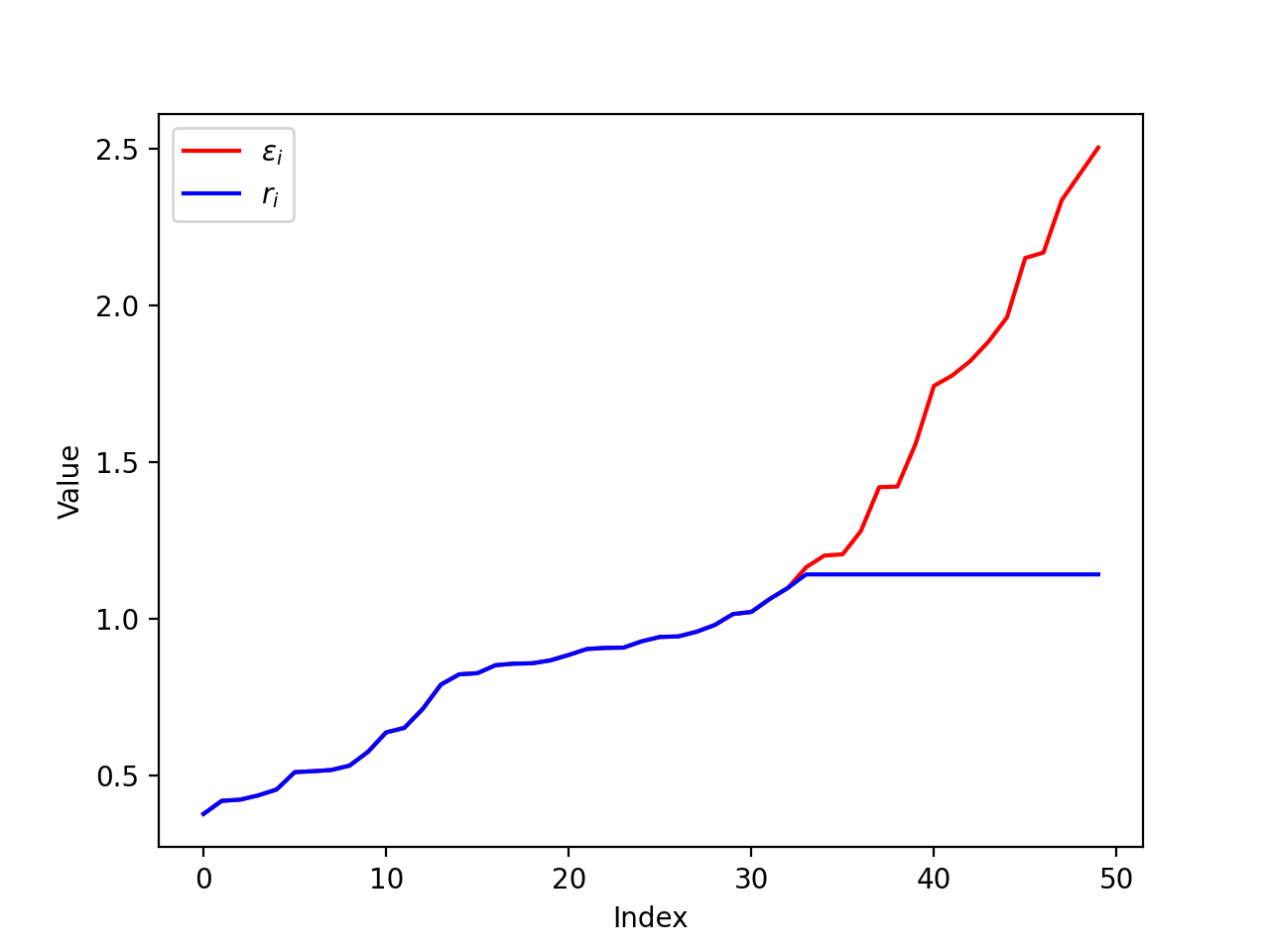}
    \caption{A specific $\{\epsilon_i\}_{i=1}^{50}$ sequence and its corresponding $\{ r_i\}_{i=1}^{50}$, given by \eqref{eq:r-def}, plotted. %
    }
    \label{fig:r-eps}
\end{figure}

Figure~\ref{fig:r-eps} shows a specific example of  $\{\epsilon_i\}_{i=1}^{50}$ and corresponding $\{ r_i\}_{i=1}^{50}$. 
After saturation, the $r_i$ values remain constant as proved in Lemma~\ref{lem:sat}(C).

Finally, we show that $\rn{n} = \argmin_{x_1^n \in \cD(\eps{n})} \frac{\Lt{x_1^n}^2 + 8}{4\Lo{x_1^n}^2}$.

\begin{lemma} \label{lem:local-min}
For a fixed $\eps{n}$, construct $\rn{n}$ as described in \eqref{eq:r-def}.
Then, $\rn{n}$ is a local optima for $f(x_1^n)$ on the domain $\cD(\eps{n})$.
\end{lemma}
\begin{proof}
Consider the partial derivative of $f(\cdot)$ with respect to $x_i$ at $x_1^n = r_1^n$.

We have 
\begin{align}
    \frac{\partial f(\rn{n})}{\partial r_i} &= \frac{r_i \Lo{\rn{n}}-\Lt{\rn{n}}^2-8}{2\Lo{\rn{n}}^3} \\
    &= \frac{Y_A + Y_B}{2\Lo{\rn{n}}^3}.
\end{align}
where $Y_A = r_i\Lo{r_{i+1}^n}-\Lt{r_{i+1}^{n}}^2$ and $Y_B = \Lo{\rn{i-1}}(r_i -4\Lo{\rn{i-1}}f(\rn{i-1}))$.

    Now consider the two cases:
    \begin{itemize}
        \item $\sn{i}$: We have $r_i = 4\Lo{\rn{i-1}}f(\rn{i-1})$ so $Y_B = 0$ and by Lemma~\ref{lem:sat}(C), $Y_A = 0$. Thus, $\frac{\partial f(\rn{n})}{\partial r_i} = 0$.
        \item $\neg \sn{i}$: We have $Y_B \leq 0$ and by \eqref{eq:r-inc}, $Y_A \leq 0$.
        Therefore, the only local perturbation that decreases the objective would need to increase $r_i$, which is not possible since $r_i = \epsilon_i$ and we have the domain restriction $\frac{r_i}{\epsilon_i} \leq \frac{r_1}{\epsilon_1} = 1$.
    \end{itemize}

    The above two cases show that $\rn{n}$ is indeed a local minimum for  $f(x_1^n)$ on the domain $\cD(\eps{n})$.
\end{proof}

\begin{proof}[\textbf{Proof of Theorem~\ref{thm:UB}}]
    Note that $f(x_1^n)$ is strictly quasi-convex on the closed convex domain $\cD(\eps{n})$ by Lemma~\ref{lem:quasi}.
Lemma~\ref{lem:local-min} shows that $\rn{n}$ is a local minimizer of the function, thus, it is also a global minimizer.
By Corollary~\ref{cor:opt}, we can use the weights $\rn{n}/\Lo{\rn{n}}$.
Note that if the upper bound on MSE, $f(\rn{n})$, is too high then one can simply output $0$ as the estimator and incur a maximum of $\frac{1}{4}$ as the MSE, as is done in the ADPM algorithm.
\end{proof}

\subsection{Interpreting ADPM} \label{sec:ADPM-inter}
ADPM exploits a crucial property of the solution of \eqref{eq:rn-defn}:
 $\vecr^* = \min_{\vecx \in \cD(\eps{n})} f(\vecx,\eps{n})$ is closely related to $\vecr' = \min_{\vecx \in \cD(\eps{n-1})} f(\vecx,\eps{n-1})$.
That is, the optimal weights when there are $n-1$ users with privacy constraints $\eps{n-1}$ and the optimal weights when there are $n$ users with privacy constraint $\eps{n}$ are closely related.
In fact, $\vecr^*_i = \vecr'_i $ for $1 \leq i \leq n-1$.
This recursive property allows us to efficiently give a solution for $\vecr^*$ given $\vecr'$, where the last component of the vector $\vecr^*$ can be found by performing a local minimization while fixing the first $n-1$ components to be $\vecr'$. 
This is precisely why we chose to work with the scaled domain $\cD(\eps{n})$ instead of the simplex.\\

From Lemma~\ref{lem:sat}, we see that upon saturation at some index $k+1$, all the subsequent indices remain saturated and
they all have the same weight $r_{k+1} = 4\Lo{\rn{k}}f(\rn{k})$ (and $\leq \epsilon_{k+1}$).
Since no saturation occured before index $k+1$, we have $\rn{k} = \eps{k}$, and thus, $r_j = \frac{\Lt{\eps{k}}^2 + 8}{\Lo{\eps{k}}} \ \forall j  \geq k+1$.
By the discussion around \eqref{eq:w-con}, using weights $\vecw = \rn{n}/\Lo{\rn{n}}$ and $\eta = \Li{\vecw/\bm\epsilon} = 1/\Lo{\rn{n}}$ gives a privacy of 
$r_i$ to user $i$. 
Thus, all the privacy-desiring users with lower $\epsilon$ requirement till user $k$ get exactly the privacy they ask for.
However, for users who want less privacy and $\epsilon \geq \epsilon_{k+1}$ receive a higher privacy guarantee of $ \frac{\Lt{\eps{k}}^2 + 8}{\Lo{\eps{k}}}$ for free.

\begin{rem}[Special Case of Two Groups of Privacy]
While the minimax optimality of ADPM is proven in this work, let us consider a special setting here to get a better intuition.
Consider the case of $n$ users where a fraction $f$ of the users all have a common privacy level $\epsilon_1$ and the rest of the users have a common privacy level $\epsilon_2$ (without loss of generality assume $\epsilon_1 \leq \epsilon_2$).
This setting was consider in \cite{isit-paper}.

The condition for saturation was found to be $\epsilon_2 \geq \epsilon_1 + \frac{8}{nf\epsilon_1}$.
It is easy to see that we recover the same condition from \eqref{eq:r-def}.
One can also verify that the weights assigned according to \eqref{eq:r-def} match the optimal weights derived in \cite{isit-paper}.
Thus, keeping $n$, $\epsilon_1$ and $f$ fixed, if $\epsilon_2$ is increased from $\epsilon_1$, then until $\epsilon_2 \leq \epsilon_1 + \frac{8}{nf\epsilon_1}$, the optimal affine estimator weighs the datapoints proportional to the privacy level.
After this saturation point, the weights do not change and latter group receives a privacy of $\epsilon_1 + \frac{8}{nf\epsilon_1}$ for free despite possibly having no privacy requirements $(\epsilon_2 \to \infty)$.
\end{rem}

\section{Lower Bound} \label{sec:LB}

In a system with $n$ users with  homogeneous differential privacy requirement $\epsilon$, the known minimax rate for mean estimation under mean-squared error is known to be $\Theta(\frac{1}{n}+\frac{1}{(n\epsilon)^2})$.
Many of the lower bound techniques in the literature for DP separately obtain the $1/n^2$ term and add the $1/n$ term by citing classical statistics results \cite{Kamath19,Duchi14}. 
Such an approach is not suitable here since the $1/n$ and the $1/n^2$ terms become intertwined by the privacy levels for heterogeneous setting.
For example, consider a special case where out of $n$ users, $n-m$ users have privacy level of $\epsilon \to 0$, and the rest $m$ users have privacy level $\epsilon \to \infty$. This corresponds to the classical mean estimation problem with $m$ samples and the mean-squared error should be of order $O(\frac{1}{m})$.
Simply adding $1/n$ to the lower bound cannot give tight results (consider the case where $m$ is a constant and $n$ is large).

We use a form of Le Cam's method adapted to differential privacy constraint to obtain a lower bound based on ideas from \cite{Duchi13,Duchi14,Duchi16}.
Our method is similar to \cite{Asu22} but it is stronger since it can handle arbitrarily large $\epsilon$ values, as is required for the case when we have a public dataset. Intuitively, DP restricts the variation in output probability with varying inputs which helps  bound the total-variation norm term in Le Cam's method.

We remind the readers that $\bm\epsilon$ is assumed to be in a non-decreasing order. 

\begin{thm}[Lower Bound] \label{thm:lb}
For privacy vector $\bm\epsilon$ of length $n$, we have
\begin{align}
L(\bm\epsilon) &\gtrsim H(\bm\epsilon) \wedge \frac{1}{4}, \\
\text{where }\ \  H(\bm\epsilon) &= \max_{i=0}^n  \frac{1}{\Lo{\eps{i}}^2 + n-i}.
\end{align}
\end{thm}
\begin{proof}
    Let $P_1$, $P_2$ be two distributions in $\cP$ and let $M$ be any $\bm\epsilon$-DP estimator of the mean, then denote the output distribution of $M(\vecX)$ with $\vecX \sim P_i^n$ as $Q_i$ for $i=1,2$.
    In other words, $Q_i$ is a distribution over $\cY$ and $Q_i(A) = \bbP_{\vecX \sim P_i^n}\{M(\vecX) \in A\}$.

Let $\delta \in [0,0.5]$; consider the distribution $P_1$ which is $0.5$ with probability $\frac{1+\delta}{2}$ and -0.5 with probability $\frac{1-\delta}{2}$. Similarly, $P_2$ is $0.5$ with probability $\frac{1-\delta}{2}$ and -0.5 with probability $\frac{1+\delta}{2}$. \\

In this case, $\mu_{P_1} = \delta/2$ and $\mu_{P_2} = -\delta/2$. Further, $\Lnorm{P_1 - P_2}{TV} = \delta$ and $\kl{P_1}{P_2} \leq 3 \delta^2 $ (for $\delta \in [0,0.5]$). 
Define $\gamma = \frac{1}{2}\left| \mu_{P_1} - \mu_{P_2} \right| = \delta/2$
, then, Le Cam's method specialized to differential privacy setting (see \cite{Duchi13,Duchi14,Duchi16}) yields the lower bound
\begin{equation} \label{eq:LeC}
L(\bm\epsilon) \geq \frac{\gamma^2}{2}(1 - \Lnorm{Q_1 - Q_2}{TV}).
\end{equation}
Using Lemma~\ref{lem:TV-bound} in the Appendix, and $1-x\leq e^{-x}\, \forall x\geq 0$, we obtain 
\begin{align} 
\Lnorm{Q_1 - Q_2}{TV} &\leq 2\delta \Lo{\eps{k}} + \delta\sqrt{\frac{3(n-k)}{2}} \quad \forall k \in \{0,\ldots,n\} \\
&\leq 2\delta \Lo{\eps{k}} + \delta\sqrt{4(n-k)} \quad \forall k  .\label{eq:QB}
\end{align}

Note that \eqref{eq:QB} holds for arbitrarily large $\epsilon_i$ values and degrades gracefully as compared to the $e^{\epsilon_i} - 1$ bound obtained in \cite[Lemma 3]{Asu22}. 
We could achieve this due to the stronger bound we derive in Lemma~\ref{lem:DP-imply} and Lemma~\ref{lem:TV-bound}. In particular, this allows us to deal with the general case when one of the datasets is public.

 Using \eqref{eq:QB} in \eqref{eq:LeC}, we obtain 
\begin{align}
 L(\bm\epsilon) &\geq \frac{\delta^2}{8}\left[1 - \delta \left(2\Lo{\eps{k}} + \sqrt{4(n-k)} \right)\right] \forall k. \label{eq:L1}
\end{align}
Setting $\delta = \frac{1}{4\Lo{\eps{k}} + 4\sqrt{(n-k)}} \wedge \frac{1}{2}$, in \eqref{eq:L1}, get 
\begin{align}
 L(\bm\epsilon) &\geq \frac{1}{16} \left( \frac{1}{16(\Lo{\eps{k}} + \sqrt{(n-k)})^2}  \wedge \frac{1}{4} \right) \forall k \\
 &\geq  \frac{1}{16} \left( \frac{1}{32(\Lo{\eps{k}}^2 + (n-k))}  \wedge \frac{1}{4} \right) \forall k \\
  &\geq  \frac{1}{512} \left( \frac{1}{\Lo{\eps{k}}^2 + n-k}  \wedge \frac{1}{4} \right) \forall k \\ 
\implies    L(\bm\epsilon) &\geq \frac{1}{512} \left( H(\bm\epsilon)  \wedge \frac{1}{4} \right) \label{eq:L2}
\end{align}
\end{proof}
\section{Optimality} \label{sec:OPT}

It remains to show that the lower and the upper bound are within constant factor of each other. 
Concretely, we prove that there exists an universal constant $c$, independent of $n$ and $\epsilon_1,\ldots,\epsilon_n$, such that
\begin{equation}
    c \left( f(\rn{n})  \wedge \frac{1}{4} \right)\leq \frac{1}{512} \left( H(\eps{n})  \wedge \frac{1}{4} \right).
\end{equation}
In the above, $\rn{n}$ is the ADPM weights defined in \eqref{eq:r-def}.
Showing the above inequality would imply 
\begin{equation}
    c \left( f(\rn{n})  \wedge \frac{1}{4} \right)\leq L(\eps{n}) \leq  \left( f(\rn{n})  \wedge \frac{1}{4} \right),
\end{equation}
proving the minimax optimality of ADPM.

Thus, it suffices to show 
$ c' f(\rn{n}) \leq H(\eps{n})$ for all $n$, and $\eps{n}$, leading to $(c' \wedge 1) \left( f(\rn{n})  \wedge \frac{1}{4} \right)\leq  \left( H(\eps{n})  \wedge \frac{1}{4} \right)$. 
In this Section, we show that this is indeed true despite $f(\rn{n})$ and $H(\eps{n})$ being expressed in rather different forms.
We show this via an indirect, recursive-like comparison of the two.
Theorem~\ref{thm:opt} states our main result on this.

It should be noted that the proof for Theorem~\ref{thm:opt} is done to show that the lower and upper bound are within constant factors of each other without care to make the constant sharp.
One can possibly get better constants with finer analysis (see Section~\ref{sec:Exp-Lem}).

\begin{thm}[Optimality] \label{thm:opt}
For any $n$ and $\eps{n}$, it holds that
\begin{equation}
    \frac{1}{443}f(\rn{n}) \leq H(\eps{n}).
\end{equation}
Therefore, we have
\begin{equation}
     \frac{1}{226816}\left( f(\rn{n}) \wedge \frac{1}{4} \right) \leq L(\eps{n}) \leq  f(\rn{n}) \wedge \frac{1}{4}.
\end{equation}
\end{thm}

The proof of Theorem~\ref{thm:opt} can be found at the end of this Section and has roughly two main parts.
Consider a privacy constraint $\eps{n}$, arranged in ascending order.
Suppose saturation first occurs at index $k+1$, i.e., $\sn{k+1}$, then in Section~\ref{sec:unsat}, we prove an algebraic result showing
$c'f(\rn{k}) \leq H(\eps{k})$, where $c'$ is independent of $k$ and $\eps{k}$.
Next, in Section~\ref{sec:Sat}, we show that if $\sn{k+1}$ and $c'f(\rn{k}) \leq H(\eps{k})$ then $c'f(\rn{k+1}) \leq H(\eps{k+1})$.
By  Lemma~\ref{lem:sat}(A), noting that saturation occurs at all the following indices, the theorem follows.

\subsection{Unsaturated Regime} \label{sec:unsat}

Consider the sequence, of length $n$, $\eps{n}$ such that no saturation occurs.
Since there is no saturation, $\rn{n} = \eps{n}$.
We prove $f(\eps{n}) \leq 443 H(\eps{n})$
for this unsaturated case. 
The result is stated in Lemma~\ref{lem:unsat-regime}.
By Definition~\ref{def:sat}, no saturation implies we have
 $\epsilon_{k-1} \leq \epsilon_{k} < 4\Lo{\eps{k-1}}f(\eps{k-1})  \ \forall \ 1 < k \leq n$. 
For the curious readers, Lemma~\ref{lem:valid} in the Appendix shows that $\epsilon_{k-1} < 4\Lo{\eps{k-1}}f(\eps{k-1})  \ \forall \ 1 < k \leq n$, thus, the above intervals are valid.

\begin{lemma} \label{lem:unsat-regime}
    If $\bigwedge_{i=2}^n (\neg \sn{i})$, then
    \begin{equation}
    \frac{1}{443}f(\eps{n}) \leq H(\eps{n}).
\end{equation}
\end{lemma}
\begin{proof}
     Recall that $H(\eps{n}) = \max_{i=0}^n  \frac{1}{\Lo{\eps{i}}^2 + n-i}$.
Observe that $H(\eps{n}) \geq 1/\Lo{\eps{n}}^2$.
Further, it is easy to see that the maximum in the definition of $H(\eps{n})$ occurs at index $p$ which is the largest such that $\epsilon_p(\epsilon_p +2\Lo{\eps{p-1}}) \leq 1$ (or at $p=0$). 
Thus, 
\begin{equation}
    \epsilon_p \Lo{\eps{p}} \leq 1 \label{eq:eps-n}
\end{equation}
unless $p=0$.
Note that if $0 < p < n$,
then we have $\epsilon_{p+1}(\epsilon_{p+1} +2\Lo{\eps{p}}) > 1$, which implies
\begin{equation}
    2\epsilon_{p+1}\Lo{\eps{p+1}} > 1. \label{eq:eps-n+1}
\end{equation} 

Consider the three cases regarding $p$.
\begin{description}
    \item[$\bullet$ $p=0$]: if $p=0$, then $\epsilon_1 \geq 1$ and $H(\eps{n}) = 1/n$. By Lemma~\ref{lem:sep0} below, we have $\frac{9^2}{4}H(\eps{n}) \geq \frac{\Lt{\eps{n}}^2}{4\Lo{\eps{n}}^2}$.
    Recall that $f(\eps{n}) = \frac{\Lt{\eps{n}}^2 + 8}{4\Lo{\eps{n}}^2}$.
Since, $2H(\eps{n}) \geq 8/4\Lo{\eps{n}}^2$, 
we get $23H(\eps{n}) \geq f(\eps{n})$.

    \item[$\bullet$ $p=n$]: if $p=n$, then by \eqref{eq:eps-n},
\begin{equation}
    \Lt{\eps{n}}^2 \leq \epsilon_n \Lo{\eps{n}} \leq 1. \label{eq:k=n}    
    \end{equation}
we have $f(\eps{n}) \leq \frac{9}{4\Lo{\eps{n}}^2} \leq 3H(\eps{n})$.

\item[$\bullet$ $1< p < n$]: This case reqires a more careful analysis. By Lemma~\ref{lem:main} below, we have $441H(\eps{n}) = \frac{441}{\Lo{\eps{p}}^2 + n-p} \geq \frac{\Lt{\eps{n}}^2}{4\Lo{\eps{n}}^2}$. Thus, $443H(\eps{n}) \geq f(\eps{n})$.
\end{description}

The above three cases combined prove the required identity in the unsaturated regime.
\end{proof}

We need Lemma~\ref{lem:unsat-ineq}, Lemma~\ref{lem:inbound}, Lemma~\ref{lem:sep0}, and Lemma~\ref{lem:bigbig} to prove Lemma~\ref{lem:main}.
Lemma~\ref{lem:unsat-ineq} gives an important inequality that is utilized in the lemmata that follow.

\begin{lemma} \label{lem:unsat-ineq}
    If $\bigwedge_{i=2}^n (\neg \sn{i})$, then
    \begin{equation} \label{eq:eps-ub}
    \epsilon_{j} \leq \epsilon_i + \frac{8}{\Lo{\eps{i}}} \ \forall \ n \geq  j > i \geq 1.
\end{equation}
\end{lemma}
\begin{proof}
From the unsaturation condition, get
 \begin{align}
    \epsilon_k  &\leq 4\Lo{\eps{k-1}}f(\eps{k-1}) \\
    \implies \epsilon_k  &\leq \frac{\Lt{\eps{k-1}}^2 + 8}{\Lo{\eps{k-1}}} \\
    \implies \epsilon_k \Lo{\eps{k-1}} &\leq \Lt{\eps{k-1}}^2 + 8, \\
\implies \sum\limits_{m=1}^{k-1} \epsilon_m(\epsilon_k - \epsilon_m) &\leq 8 \\
\implies \sum\limits_{m=1}^{k} \epsilon_m(\epsilon_k - \epsilon_m) &\leq 8  \ \ \forall \ 1 < k \leq n.
 \end{align}
As a consequence, for any $j > i \geq 1$, 
\begin{align}
 (\epsilon_{j} - \epsilon_i)\sum \limits_{m=1}^i \epsilon_m &\leq 8   \\
 \implies \epsilon_{j} &\leq \epsilon_i + \frac{8}{\Lo{\eps{i}}}. 
\end{align}
\end{proof}

Now we prove two lemmata that come in use to prove Lemma~\ref{lem:bigbig}.

\begin{lemma}\label{lem:inbound}
Suppose $\epsilon_n \leq C\epsilon_i$ for some $C > 0$ and $i$, then 
$$C^2 \Lt{\ep{n}{i}}^2 \geq (n-i+1)\Lt{\ep{n}{i}}^2.$$
\end{lemma}
\begin{proof}
    $$C^2 \Lt{\ep{n}{i}}^2 \geq (n-i+1)^2\epsilon_n^2 \geq (n-i+1)\Lt{\ep{n}{i}}^2.$$
\end{proof}

\begin{lemma}\label{lem:sep0}
Suppose $\epsilon_i > C$ for some $C > 0$ and $i$, then 
$$(1 + 8C^{-2})^2 \Lt{\ep{n}{i}}^2 \geq (n-i+1)\Lt{\ep{n}{i}}^2.$$
\end{lemma}
\begin{proof}
    From \eqref{eq:eps-ub}, $\epsilon_n \leq \epsilon_i + \frac{8}{C} \leq \epsilon_i(1 + 8C^{-2})$. The result then follows from Lemma~\ref{lem:inbound}.
\end{proof}

Before going to Lemma~\ref{lem:main}, we state and prove another lemma.

\begin{lemma}\label{lem:bigbig}
Suppose that for some $i$, $\epsilon_{i} \Lo{\ep{i}{1}} \leq 1$ and $2\epsilon_{i+1}\Lo{\ep{1}{i+1}} > 1$, then 
$$42^2 \Lo{\ep{n}{i+1}}^2 \geq (n-i)\Lt{\ep{n}{i+1}}^2$$
\end{lemma}
\begin{proof}
    The second inequality in the statement of the lemma implies that 
    $$\epsilon_{i+1} > \frac{1}{2}\Lo{\epsilon_1^{i+1}}^{-1}.$$ 
    On the other hand, \eqref{eq:eps-ub} gives us $\epsilon_{n} \leq \epsilon_{i+1} + 8\Lo{\epsilon_1^{i+1}}^{-1}$. 
    Thus, 
    \begin{align}
        \frac{\epsilon_n}{\epsilon_{i+1}} &< 2\frac{\epsilon_{i+1} + 8\Lo{\epsilon_1^{i+1}}^{-1}}{\Lo{\epsilon_1^{i+1}}^{-1}} \\
        &= 2(8 + \epsilon_{i+1}\Lo{ \epsilon_1^{i+1}}). \label{eq:s0}
    \end{align}
   
    Note that 
    \begin{align}
    \epsilon_{i+1}\Lo{ \epsilon_1^{i+1}} & = \epsilon_{i+1}(\epsilon_{i+1} + \Lo{\epsilon_1^{i}}) \\
    &= \epsilon_{i+1}^2 + (\epsilon_{i+1} - \epsilon_i)\Lo{\epsilon_1^{i}} + \epsilon_i \Lo{\epsilon_1^{i}} \\ 
    & \leq \epsilon_{i+1}^2 + (\epsilon_{i+1} - \epsilon_i)\Lo{\epsilon_1^{i}}  + 1. \label{eq:s1}
    \end{align}
    From \eqref{eq:eps-ub} we know that 
    \begin{equation}
    \epsilon_{i+1} - \epsilon_i \leq 8\Lo{\epsilon_1^{i}} ^{-1}. \label{eq:s2}    
    \end{equation}
    Combining \eqref{eq:s1} and \eqref{eq:s2} into \eqref{eq:s0}, we get 
    $$ \frac{\epsilon_n}{\epsilon_{i+1}} < 2(8 + \epsilon_{i+1}^2 + 9) = 34 + 2\epsilon_{i+1}^2.$$
    Now if $\epsilon_{i+1} \leq 2$, the result follows from Lemma~\ref{lem:inbound}. Otherwise it follows from Lemma~\ref{lem:sep0}.
\end{proof}

We now prove the main lemma.

\begin{lemma}\label{lem:main}
    $$ 42^2  \Lo{\epsilon_1^n}^2 \geq  \Lt{\epsilon_1^n}^2(n - p +  \Lo{\epsilon_1^p}^2). $$
\end{lemma}
\begin{proof}
    From \eqref{eq:eps-n},
    \begin{equation}
    \Lt{\epsilon_1^{p}}^2 \leq \epsilon_p \Lo{\epsilon_1^{p}} \leq 1. \label{eq:k=n}    
    \end{equation}
    Thus, 
    \begin{equation}
    \Lo{\epsilon_1^{p}}^2 \geq \Lt{\epsilon_1^{p}}^2 \cdot \Lo{\epsilon_1^{p}}^2. \label{eq:e3} 
    \end{equation}
    Applying Lemma~\ref{lem:bigbig} we also have 
    \begin{equation}42^2\Lo{\epsilon_{p+1}^{n}}^2 \geq (n-p)\Lt{ \epsilon_{p+1}^{n}}^2. \label{eq:e4} \end{equation}
    Further, observe that
    \begin{equation}
    \Lo{\epsilon_{1}^{p}} \cdot \Lo{\epsilon_{p+1}^{n}} \geq (n-p)\Lt{\epsilon_{1}^{p}}^2. \label{eq:e5} 
    \end{equation} 
    since $\{\epsilon_i\}$ is non-decreasing.
    From \eqref{eq:eps-ub}, note that for all $i > p$ we have 
    $$ \epsilon_i \Lo{\epsilon_{1}^{p}} \leq \epsilon_p \Lo{\epsilon_{1}^{p}} + 8 \leq 9.$$
    Multiplying by $\epsilon_i$ and adding over all $i > p$ we get
    \begin{align}
    9\Lo{\epsilon_{p+1}^{n}} &\geq \Lt{\epsilon_{p+1}^{n}}^2 \cdot \Lo{\epsilon_{1}^{p}} \\
    \implies 9\Lo{\epsilon_{p+1}^{n}}\Lo{\eps{p}} &\geq \Lt{\epsilon_{p+1}^{n}}^2 \cdot \Lo{\epsilon_{1}^{p}}^2. \label{eq:e6}
    \end{align}
    
    From \eqref{eq:e3}, \eqref{eq:e4}, \eqref{eq:e5} and \eqref{eq:e6} (add them and upper bound the upper bound), we get 
    $$ 42^2 \Lo{\epsilon_1^n}^2 \geq \Lt{\epsilon_1^n \rVert}^2(n - p + \Lo{\epsilon_1^p}^2).$$
\end{proof}

\subsection{Saturated Regime} \label{sec:Sat}

In the saturated regime, we show that if $\sn{k+1}$ occurs and we have
$H(\eps{k}) \geq \frac{1}{443} f(\rn{k})$, then $H(\eps{k+1}) \geq \frac{1}{443} f(\rn{k+1})$.

\begin{lemma} \label{lem:rec}
If $\sn{k+1}$ and $H(\eps{k}) \geq \frac{1}{443} f(\rn{k})$, then $H(\eps{k+1}) \geq \frac{1}{443} f(\rn{k+1})$.
\end{lemma} 
\begin{proof}
Suppose at index $p$, $H(\eps{k})$ was maximized, i.e.,
$H(\eps{k}) = \frac{1}{\Lo{\eps{p}}^2 + k-p}$.
Then,
\begin{align}
  H(\eps{k+1}) &\geq  \frac{1}{\Lo{\eps{p}}^2 + k+1-p}  \\
  &= \frac{H(\eps{k})}{1+H(\eps{k})}
\end{align}
Noting that $x/(1+x)$ is increasing function in $x$ and $H(\eps{k})$ is lower bounded by $\frac{1}{443} f(\rn{k})$, we have,
\begin{align}
H(\eps{k+1}) &\geq  \frac{ f(\rn{k})}{443+f(\rn{k})}  \\
 &\geq  \frac{1}{443} \frac{f(\rn{k})}{1+4 f(\rn{k})} \\
 &= \frac{1}{443} f(\rn{k+1})
\end{align}
where we used $f(\rn{k+1}) =\frac{f(\rn{k})}{1+4 f(\rn{k})}$ due to $\sn{k+1}$ (see Lemma~\ref{lem:sat}(B)).
\end{proof}

\begin{proof}[\textbf{Proof of Theorem~\ref{thm:opt}}]
Lemma~\ref{lem:unsat-regime} and Lemma~\ref{lem:rec} together prove Theorem~\ref{thm:opt}:
 suppose for $\eps{n}$, $k$ is the least index at which $\sn{k+1}$ occurs, then by Lemma~\ref{lem:unsat-regime}, since $\eps{k}$ is an unsaturated sequence, $H(\eps{k}) \geq \frac{1}{443}f(\rn{k})$.
By Lemma~\ref{lem:sat}(A), all indices from $k$ onward are saturated, so applying Lemma~\ref{lem:rec} $n-k$ times for all subsequent indices result in $H(\eps{n}) \geq \frac{1}{443}f(\rn{n})$.    
\end{proof}
\section{Experiments} \label{sec:Exp}

\subsection{Baseline Schemes}

We first describe several  baseline DP techniques and discuss why they are not optimal in HDP.  
Supporting experiments follow.

\textbf{Uniformly enforce $\epsilon_1$-DP (UNI):} One approach to this problem is to offer $\epsilon_1$ (the lowest value in $\bm\epsilon$) privacy to all the datapoints and use the minimax estimator, i.e., sample mean added with Laplace noise, to get an error of $O(1/n + 1/(n\epsilon_1)^2)$. 
UNI can be arbitrarily worse than the ADPM. 
Consider the case when only 1 datapoint has an extremely stringent privacy requirement while rest of the data is public. 

\textbf{Sampling Mechanism (SM) \cite{Jorg15}:} 
Let $t = \Lnorm{\bm\epsilon}{\infty}$, then sample the $i$-th datapoint independently with probability $(e^{\epsilon_i}-1)/(e^t - 1)$ and apply any homogeneous $t$-DP algorithm on the sub-sampled dataset.
\cite{Jorg15} proved this mechanism is $\bm\epsilon$-DP.
For our case, take the sample mean of the sub-sampled dataset and add Laplace noise with variance $2/(N_s t)^2$, where $N_s$ is the realization of the number of datapoints sub-sampled.
However, this approach can be easily seen to be suboptimal.
When one datapoint is public, the SM algorithm disregards rest of the private data.

\textbf{Local Differential Private Estimator (LDPE):}  
For lack of a good estimator for HDP, we also consider a Local-DP estimator.
Each user with privacy requirement $\epsilon_i$ adds Laplace($1/\epsilon_i$) noise to their data and sends it to the central server.
The central server produces a weighted combination of the noisy data as the estimator.
We can take optimal linear combinations of these noisy values to minimize the mean squared error if the variance of the unknown distribution is known (see \cite{Ferrando21} for details).
We take the worst case variance as a proxy in our problem setting.
In general, Local-DP setting adds more noise to get privacy from the central serve - this is a known shortcoming of the Local-DP model so we should not expect this LDPE model to be on par with the other baseline techniques.

\textbf{ Fallah et al.'s mean estimator (FME) \cite{Asu22}:} For brevity, we direct the readers to \cite[Theorem 1]{Asu22} for details on the algorithm.
We refer to this algorithm as FME in the rest of this work.
One of the shortcomings of this method is it assumes $\Lnorm{\bm\epsilon}{\infty} \leq 1$ for its theoretical guarantees.
For our experiments, we still use this algorithm as it is stated for $\Lnorm{\bm\epsilon}{\infty} > 1$.
Even when $\Lnorm{\bm\epsilon}{\infty}  \leq  1$, FME may perform orders of magnitude worse than ADPM (see Table~\ref{tab:het-dp}). 
Further, if the example in Remark~\ref{rem:prop} is considered, then FME obtains worst-case MSE of the order $10^{-3}$ as compared to ADPM's MSE of $10^{-4}$.

\textbf{Proportional DP (PropDPM):} We refer to the affine estimator with weights proportional to the $\bm\epsilon$ vector and appropriate Laplace noise as PropDPM.
This estimator also suffers from the problem of disregarding private data if even one of the datapoints is public, as pointed out in Remark~\ref{rem:prop}.

\subsection{Experiments}

We run two types of experiments for comparing ADPM to the other algorithms under heterogeneous DP constraints.
We consider two cases for $\bm\epsilon$ of dimension  $n=10^3$: high variance and low variance in $\bm\epsilon$.
The low variance case is obtained by uniformly sampling $\log \bm\epsilon$ in $[-3,-2]$.
Independently, the high variance case corresponds to sampling  $\log \bm\epsilon$ in $[-4,2]$. 
Keeping the sampled $\bm\epsilon$ fixed, the average of the squared errors was taken over 20K simulations under Beta$(2,3)$ distribution on $\cX$. 
The results are presented in Table~\ref{tab:het-dp}.
It is not surprising that UNI, PropDPM and ADPM enjoy similar  performance in the low variance regime, while they diverge in the higher variance regime.
LDPE performing poorly in the low variance regime is also not surprising as explained earlier.
However, in high variance regime, LDPE performs decently.
It should be noted that in the low variance regime, the realization of $\bm\epsilon$ satisfied $\Li{\bm\epsilon} \leq 1$, the condition required in FME. 
However, FME is still 2 orders of magnitude worse than ADPM.

\begin{table}
\caption{Comparison of MSE of six methods for high and low variance in $\bm\epsilon$.}
\label{tab:het-dp}
\centering
\begin{tabular}{|c|c|c|} \hline
\textbf{Method} &
\begin{tabular}[c]{@{}c@{}} \boldmath$\log$ \textbf{MSE} \\ High $\mathsf{Var}(\bm\epsilon)$\end{tabular} &
\begin{tabular}[c]{@{}c@{}} \boldmath$\log$ \textbf{MSE} \\ Low $\mathsf{Var}(\bm\epsilon)$\end{tabular} \\ \hline 
ADPM & -\textbf{9}.\textbf{3} & -\textbf{8}.\textbf{1} \\  \hline
PropDPM & -9.0 & -8.1 \\ \hline
LDPE & -7.2  & -1.3\\ \hline
SM & -6.5 & -7.9\\ \hline
FME & -6.2 & -6.2 \\ \hline
UNI & -5.1 & -7.1 \\ \hline
\end{tabular}
\end{table}

\subsection{On Lemma~\ref{lem:unsat-regime}} \label{sec:Exp-Lem}

In order to check the tightness of the constant derived in Lemma~\ref{lem:unsat-regime}, we scatter plot various values of $\log H(\eps{n})$ and $\log f(\rn{n})$ in Figure~\ref{fig:scat}.
$n$ is randomly chosen, and $\epsilon_{i+1}$ is sampled from $ [\epsilon_i,4
\Lo{\eps{i}}f(\eps{i}))$ for $1 < i \leq n$.
The resulting values of $\log H(\eps{n})$ and $\log f(\rn{n})$ is shown in a scatter-plot.  
The figure also includes lines showing $y = x - \log c$ for $c \in \{4,443\}$.
While Figure~\ref{fig:scat} is not conclusive, it empirically suggests the constant $443$ in  Lemma~\ref{lem:unsat-regime} can  be improved to $4$.

\begin{figure}
    \centering
    \includegraphics{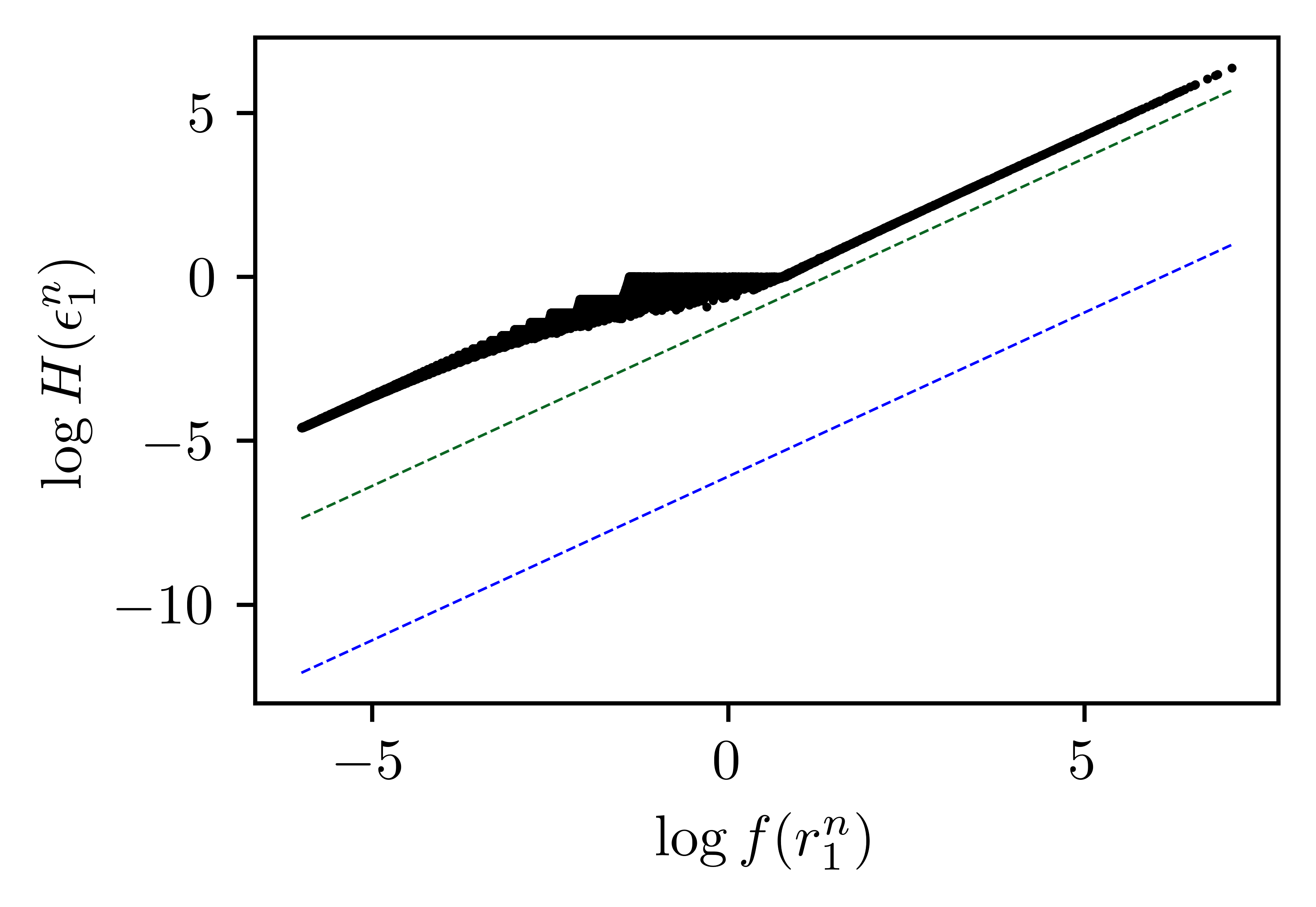}
    \caption{Scatter Plot of $\log H(\eps{n})$ vs $\log f(\rn{n})$. The dotted green line plots $y = x - \log 4$ and the dotted blue line plots $y = x - \log 443$.}
    \label{fig:scat}
\end{figure}

\section{Conclusion} \label{sec:Con}

Bounded univariate mean estimation under heterogeneous privacy constraints is studied under the Central-DP model.
We propose an efficient algorithm (ADPM) and prove the minimax optimality of ADPM up to constant factors.
In order to do this, we provide a recursive solution to an optimization problem and a lower bound.
We further show the order-equivalence of the lower and the upper bound.
Experimentally, we confirm the superior performance of our algorithm.
Further research directions can include studying the problem for sub-Gaussian random variables and the multivariate case.

\section*{Acknowledgments}
The authors thank Yigit Efe Erginbas, and Justin Singh Kang for their valuable insights.

\appendix
\section{Lemmas} \label{ap:A}
Lemma~\ref{lem:eps-af} is included for completeness (see \cite{Dwork14Alg,Asu22}).
\begin{lemma}[Laplace Mechanism] \label{lem:eps-af}
The affine estimator $M_{\vecw}(\vecx) = \inprod{\vecx}{\vecw} + L(\eta)$ is $(\vecw/\eta)$-DP when $\cX = [-0.5,0.5]$.
\end{lemma}
\begin{proof}
We verify this by comparing the density of output of the algorithm on neighboring datasets. We also drop the subscript for the algorithm $M_{\vecw}$.

\begin{align}
    \frac{p(M(\vecx) = s)}{p(M(\vecx'^i) = s)} &=  \frac{\exp\{-|\inprod{\vecx}{\vecw}-s|/\eta\}}{\exp\{-|\inprod{\vecx'^i}{\vecw} - s|/\eta\}}, \\
    &\leq \exp\{|\inprod{\vecx'^i}{\vecw}-\inprod{\vecx}{\vecw}|/\eta\}, \\
    &\leq \exp\{w_i/\eta\} \label{eq:Lap},
\end{align}
where \eqref{eq:Lap} follows from the fact that $\inprod{\vecx'^i}{\vecw}$ and $\inprod{\vecx}{\vecw}$ can differ by at most $w_i$ since $w_i$ is the coefficient of the $i$-th element and $x_i,x'^i_i \in [-0.5,0.5]$.
\end{proof}

\begin{lemma} \label{lem:DP-imply}
By the definition of $\bm\epsilon$-DP in \eqref{eq:DP-def}, it follows that for all measurable sets $S \subseteq \cY$,
\begin{align} \label{eq:DP-corr}
    e^{-\epsilon_i} \pr{M(\vecx'^i) & \in S}
     \leq \pr{M(\vecx) \in S} \\
     &\leq \begin{cases}
    e^{\epsilon_i} \pr{M(\vecx'^i) \in S}\\
    1-e^{-\epsilon_i} + e^{-\epsilon_i}\pr{M(\vecx'^i) \in S} \end{cases}
\end{align}
Further, it follows that
\begin{equation}
   | \pr{M(\vecx) \in S} - \pr{M(\vecx'^i) \in S} | \leq 1-e^{-\epsilon_i} \label{eq:Delta-Bound} 
\end{equation}
\begin{proof}
Note that the DP definition also implies $e^{-\epsilon_i} \pr{M(\vecx'^i) \in S}
     \leq \pr{M(\vecx) \in S}$.
By applying the definition with $S^C$ and combining the conditions, one obtains
\[
 \begin{rcases}
    e^{-\epsilon_i} \pr{M(\vecx'^i) \in S}\\
    1-e^{\epsilon_i} +e^{\epsilon_i}\pr{M(\vecx'^i) \in S}
    \end{rcases}
     \leq \pr{M(\vecx) \in S}  \]
and 
\[ \pr{M(\vecx) \in S} 
\leq \begin{cases}
    e^{\epsilon_i} \pr{M(\vecx'^i) \in S}\\
    1-e^{-\epsilon_i} + e^{-\epsilon_i}\pr{M(\vecx'^i) \in S}
\end{cases}\]
The condition $1-e^{\epsilon_i} +e^{\epsilon_i}\pr{M(\vecx'^i) \in S} \leq \pr{M(\vecx) \in S}$ can be removed since $e^{-\epsilon_i} \lambda \geq 1-e^{\epsilon_i} +e^{\epsilon_i}\lambda$ for all $0 \leq \lambda \leq 1$ and non-negative $\epsilon_i$. 
The Lemma gives a much stronger bound than the straightforward DP definition when $\epsilon_i$ is large. Using \eqref{eq:DP-corr}, one can obtain \eqref{eq:Delta-Bound}.
\end{proof}
\end{lemma}

\begin{lemma} \label{lem:TV-bound}
For any $k \in \{0,1,\ldots,n\}$,
\begin{align}
    \Lnorm{Q_1 - Q_2}{TV} \leq 2\Lnorm{P_1 & - P_2}{TV} \sum_{i=1}^k(1-e^{-\epsilon_i}) \\ 
    &+ \sqrt{\frac{n-k}{2}\kl{P_1}{P_2}}.
\end{align}
 
\begin{proof}
We use a method similar to \cite{Asu22} to prove this but obtain stronger results due to Lemma~\ref{lem:DP-imply}. Let $\tilde Q$ be the distribution of the output of $\bm\epsilon$-DP estimator $M(\cdot)$ when the input is dataset $\vecX$ draw from the product distribution $P_1^kP_2^{n-k}$. 
By triangle inequality,
\begin{equation}
   \Lnorm{Q_1 - Q_2}{TV} \leq \Lnorm{Q_1 - \tilde Q}{TV} + \Lnorm{\tilde Q - Q_2}{TV}. 
\end{equation}

By Data Processing Inequality and Pinsker's inequality, 
\begin{equation}
   \Lnorm{Q_1 - \tilde Q}{TV} \leq  \Lnorm{P_1^n - P_1^kP_2^{n-k}}{TV} \leq \sqrt{\frac{n-k}{2}\kl{P_1}{P_2}}. 
\end{equation}

We remind the readers that in DP-literature, the expression $\pr{M(\vecX) \in A}$ generally refers to the conditional probability $\pr{M(\vecX) \in A|\vecX}$ as pointed out in Definition~\ref{def:epsDP}.
With this convention, consider the term $\Lnorm{\tilde Q - Q_2}{TV}$,
\begin{align}
    |\tilde Q(A) &- Q_2(A)| =  | \bbE_{\vecX \sim P_1^kP_2^{n-k}}\pr{M(\vecX) \in A}  \\
    &\quad \quad \quad \quad \quad - \bbE_{\vecX \sim P_2^{n}}\pr{M(\vecX) \in A} | \\
    &= \Big| \sum_{i=1}^k \big( \bbE_{\vecX \sim P_1^iP_2^{n-i}}\pr{M(\vecX) \in A} \\ &\quad \quad - \bbE_{\vecX \sim P_1^{i-1}P_2^{n-i+1}}\pr{M(\vecX) \in A} \big) \Big| \\
    &\leq \sum_{i=1}^k \big|\bbE_{\vecX \sim P_1^iP_2^{n-i}}\pr{M(\vecX) \in A} \\ 
    &\quad \quad \quad - \bbE_{\vecX \sim P_1^{i-1}P_2^{n-i+1}}\pr{M(\vecX) \in A} \big| .
\end{align}
Let $\vecX'^i$ denote another dataset which differs from $\vecX$ at only the $i$-th index and this index can have any arbitrary value independent of $\vecX$. Then, note that 
\begin{align}
\bbE_{\vecX \sim P_1^iP_2^{n-i}}&\pr{M(\vecX'^i) \in A} \\
&- \bbE_{\vecX \sim P_1^{i-1}P_2^{n-i+1}}\pr{M(\vecX'^i) \in A} = 0 \ \ \ a.s.
\end{align}
Thus,
\begin{align}
    &|\tilde Q(A) - Q_2(A)| \leq  \\
    &\sum_{i=1}^k \Big|\bbE_{\vecX \sim P_1^iP_2^{n-i}}\big[\pr{M(\vecX) \in A}-\pr{M(\vecX'^i) \in A}\big]\\
    &-\bbE_{\vecX \sim P_1^{i-1}P_2^{n-i+1}}[\pr{M(\vecX) \in A}-\pr{M(\vecX'^i) \in A}]  \Big|. \label{eq:ExpDiff}
\end{align}
Let $\vecX_{-i}$ denote the random vector $\vecX$ except at the $i$-th position, i.e., $\vecX_{-i} \sim P_1^{i-1}P_2^{n-i}$ means that elements of $\vecX$ from position $1$ to position $i-1$ are iid $P_1$ while those in positions $i+1$ to $n$ are iid $P_2$. 

Therefore, we get,
\begin{align}
    &|\tilde Q(A) - Q_2(A)| \leq \\
    &\sum_{i=1}^k \Bigg|\bbE_{\vecX_{-i} \sim P_1^{i-1}P_2^{n-i}}\Big[\bbE_{X_i \sim P_1}[\pr{M(\vecX) \in A} \\
    &\hspace{130pt}  -\pr{M(\vecX'^i) \in A}]  \\
     &\quad\quad -  \bbE_{X_i \sim P_2}[\pr{M(\vecX) \in A}-\pr{M(\vecX'^i) \in A}]\Big]  \Bigg| \\ 
    &\leq  \sum_{i=1}^k \bbE_{\vecX_{-i} \sim P_1^{i-1}P_2^{n-i}}\Big| \bbE_{X_i \sim P_1}[\pr{M(\vecX) \in A}\\ 
    &\hspace{130pt} -\pr{M(\vecX'^i) \in A}]  \\
     &\quad \quad - \bbE_{X_i \sim P_2}[\pr{M(\vecX) \in A}-\pr{M(\vecX'^i) \in A}]\Big| \\
    &\leq  \sum_{i=1}^k \bbE_{\vecX_{-i} \sim P_1^{i-1}P_2^{n-i}}[ 2(1-e^{-\epsilon_i})\Lnorm{P_1 - P_2}{TV}]  \label{eq:Jump} \\
    &= 2\Lnorm{P_1 - P_2}{TV}\sum_{i=1}^k(1-e^{-\epsilon_i}).
\end{align}

In \eqref{eq:Jump}, we used Lemma~\ref{lem:DP-imply} and the fact that for a bounded function $|f(x)| \leq C$, we have $\left|E_{X\sim P_1}[f(X)] - E_{X\sim P_2}[f(X)] \right| \leq 2\Lnorm{P_1 - P_2}{TV}C$.

\end{proof}
\end{lemma}

\begin{lemma} \label{lem:valid}
    We have $\epsilon_1 < 4\epsilon_1f(\epsilon_1)$ and
    if $\epsilon_{k-1} \leq \epsilon_k < 4\Lo{\eps{k-1}}f(\eps{k-1})$, then $\epsilon_k < 4\Lo{\eps{k}}f(\eps{k})$.
\end{lemma}
\begin{proof}
    Note that $4\epsilon_1f(\epsilon_1) = \epsilon_1 + \frac{8}{\epsilon_1} > \epsilon_1$. 
    We also have
    \begin{align}
        \epsilon_k &< 4\Lo{\eps{k-1}}f(\eps{k-1}) \\
        \implies \epsilon_k \Lo{\eps{k-1}} &< \Lt{\eps{k-1}}^2 + 8 \\
        \implies \epsilon_k \Lo{\eps{k}} &< \Lt{\eps{k}}^2 + 8 \\
        \epsilon_k &< 4\Lo{\eps{k}}f(\eps{k}).
    \end{align}
    
\end{proof}

\bibliographystyle{IEEEtran}
\bibliography{Ref}

\vfill

\end{document}